\documentclass[conference]{IEEEtran}

\usepackage{diagbox}
\usepackage{tikz}
\usepackage{stfloats} 
\UseRawInputEncoding
\usepackage{cite}
\usepackage{amsmath}
\usepackage{amsfonts}
\usepackage{graphicx}
\usepackage{float}

\usepackage[table]{xcolor}
\definecolor{mygreen}{RGB}{0,150,30}
\usepackage{epsfig}
\usepackage{amsthm}
\usepackage{mathrsfs}
\usepackage{amsfonts}
\usepackage{amssymb}
\usepackage{mathrsfs}
\usepackage{euscript}
\usepackage{multirow}
\usepackage{url}
\usepackage{subcaption}
\usepackage{algorithm}
\usepackage{algorithmic}
\usepackage{bm}
\usepackage{mathtools} 

\def\LRT#1#2{\!
\raisebox{.2ex}{$
{{\scriptstyle\;#1}\atop{\displaystyle\gtrless}}
\atop
{\raisebox{-1.25ex}{$\scriptstyle\;#2$}}
$}
\!}

\usepackage{xcolor, soul}
\sethlcolor{green}
\def\BibTeX{{\rm B\kern-.05em{\sc i\kern-.025em b}\kern-.08em
    T\kern-.1667em\lower.7ex\hbox{E}\kern-.125emX}}

\theoremstyle{remark}
\theoremstyle{definition}\newtheorem{theorem}{Theorem}

\newtheorem*{remark*}{Remark}

\let\oldproofname=\proofname 
\renewcommand{\proofname}{\rm\bf{\oldproofname}}

\belowdisplayskip=\abovedisplayskip
\newcommand\submittedtext{%
  \footnotesize This work has been submitted to the IEEE for possible publication. 
  Copyright may be transferred without notice, after which this version may no longer be accessible.}

\newcommand\submittednotice{%
\begin{tikzpicture}[remember picture,overlay]
  \node[anchor=north,yshift=-10pt] at (current page.north)
    {\fcolorbox{red}{white}{%
      \parbox{\dimexpr0.9\textwidth-2\fboxsep-2\fboxrule}{%
    \raggedright\submittedtext}}};
\end{tikzpicture}%
}

\AddToHook{shipout/firstpage}{\submittednotice}

\begin{document}

\bstctlcite{IEEEexample:BSTcontrol}

\title{Target Detection with Tightly-coupled Antennas: Analysis for Unknown Wideband Signals}
\author{Erfan Khordad*, Peter J. Smith$^\dagger$, Sachitha C. Bandara*, Rajitha Senanayake* and Robert W. Heath Jr$^\dagger$$^\dagger$\\
*the Department of Electrical and Electronic Engineering, University of Melbourne, Melbourne, Australia\\
$\dagger$ School of Mathematics and Statistics,
Victoria University of Wellington, Wellington, New Zealand \\
$\dagger \dagger$ Department of Electrical and Computer Engineering, University of California at San Diego, San Diego, USA}
\maketitle


\begin{abstract}
This paper presents analysis for target detection using tightly-coupled antenna (TCA) arrays with high mutual coupling (MC). We show that the wide operational bandwidth of TCAs is advantageous for target detection. We assume a sensing receiver equipped with a TCA array that collects joint time and frequency samples of the target's echo signals. Echoes are assumed to be unknown wideband signals, and noise at the TCA array follows a frequency-varying correlation model due to MC. We also assume that the echo signals are time varying, with no assumption on the temporal variation. We consider three regimes in frequency as constant, slowly or rapidly varying, to capture all possible spectral dynamics of the echoes. We propose a novel detector for the slowly-varying regime, and derive detectors based on maximum likelihood estimation (MLE) for the other regimes. For the rapidly-varying regime, we derive an extended energy detector for correlated noise with frequency and time samples. We analyze the performance of all the detectors. We also derive and analyze an ideal detector giving an upper bound on performance. We validate our analysis with simulations and demonstrate that our proposed detector outperforms the MLE-based detectors in terms of robustness to frequency variation. Also, we highlight that TCA arrays offer clear advantages over weakly-coupled antenna arrays in target detection.

\end{abstract}


\section{Introduction}\label{IntroductionLABEL}

For beyond-5G wireless communications systems, using a substantial number of antennas in an array is imperative to achieve high spectral efficiency. Hence, antennas are packed together, resulting in antenna spacing below half-wavelength, and therefore mutual coupling (MC) in such arrays is no longer negligible\cite{SWmassiveMIMO,Bandara_2025_multiuser}. Ultra-dense arrangements of the antennas can also form quasi-continuous arrays \cite{RicianContinuousAntenna,MUConAper10158997,Continuous10570543}.

MC has been widely regarded as a cause of performance degradation and its effect has been mitigated, e.g. \cite{Jiongpei450}. However, the pioneering works such as \cite{Munk2006,MunkBook2003} showed that high MC is advantageous. Tightly-coupled antennas (TCAs) with high MC yield a broad operational bandwidth, at the expense of increased design complexity. The wideband gain was shown by implementations \cite{PracticalWidebandObserved,Cavallo6374210} and physically-consistent modeling \cite{SWmassiveMIMO,SachithaPaper,Bandara_2025_multiuser}. Physically-consistent models use circuit theory to link physics with communication theory\cite{CircuitTheory12}. 



Target detection is one of the areas that can significantly benefit from the merits of TCAs, as we show in this paper. Additionally, enhanced target detection using TCAs can also be beneficial for integrated sensing and communications (ISAC). Target detection was conducted for ISAC setups, e.g. in \cite{MultiTargetMultiU045,clutter890}, however without exploring the advantages of high MC. 

Target detection, as an independent problem, was studied for a range of system models, for example in \cite{RadarUnknownMutual12} by assuming unknown MC, yet considering MC as an adverse factor. Target detection with unknown wideband signals was performed using multiple energy detectors in sub-bands \cite{MultipleNarrowB_ED}. However, energy detection is susceptible to noise uncertainty and is sub-optimal for correlated signals \cite{ED_notOptimal}. Wideband detection was also considered in \cite{TimeReversal320} and in \cite{ED_notOptimal,Peng11095} detection was investigated by taking into account signal correlation models. For the multiple-antenna receiver in \cite{ED_notOptimal} only independent and identically distributed noise was considered. Also a single antenna receiver was considered in \cite{Peng11095} along with the presentation of asymptotic results. Furthermore, \cite{Peng11095,TimeReversal320} focused solely on white noise. All these methods performed detection in one domain. Target detection using time and frequency samples was performed in \cite{Kalra8635078,Javed5709749,Sattar892667}, however without detection performance analysis. Despite the above efforts, the merits of TCAs in enhancing target detection are still unexplored.

In this paper, we show the detection gains resulting from employing TCAs. We assume that the signals reflected from the target are unknown and wideband. Thus our approach is applicable to any wideband waveform and can be utilized in both communications-centric and radar-centric ISAC waveform designs. We assume time-varying echo signals that fall into three regimes in the frequency domain as constant, slowly or rapidly varying, with no specific assumption on the rate of time variation. Also, we factor in the frequency-varying noise correlation caused by MC. We propose a novel detector for the slowly-varying regime and derive MLE-based detectors for other regimes. Our detector for the rapidly-varying regime is an extended energy detector for frequency-varying correlated noise, which collects time and frequency samples. 
We further derive a detector that gives an upper bound on detection performance. We present rigorous performance analysis by calculating the distribution of all the detectors under null and alternative hypotheses. We show perfect accuracy of our analysis by simulations and demonstrate that our detector has greater robustness than MLE-based detectors. We also show that TCA arrays exhibit clear benefits over conventional weakly-coupled antenna arrays in target detection.

\section{System Model and Problem Formulation}\label{Sec:SystemModel}

Following \cite{SWmassiveMIMO}, we present a physically-consistent system model using circuit theory, assuming canonical minimum scattering (CMS) antennas, which are suitable for wideband applications, to model TCAs\cite{CMS_Antennas1965}. A CMS antenna has an equivalent resistor-inductor-capacitor (RLC) circuit, based on Chu's theory. We call a CMS antenna modeled based on Chu's theory a Chu's CMS antenna. The radius of a sphere that encloses the Chu's CMS antenna is denoted by $a_R$. Impedance parameters of the RLC circuit are a function of $a_R$.

Equivalence between Hertz dipoles with Chu's CMS antennas, using radiated power, gives an expression for the mutual impedance of two Chu's CMS antennas\cite[Eq. (5)]{SWmassiveMIMO}. This expression considers the antenna spacing denoted by $\delta$ and the arrangements of the dipoles. We use a colinear arrangement in which the dipoles are placed along a straight line\cite{SWmassiveMIMO,SachithaPaper}.

We consider a transmitter that transmits wideband signals. These signals reflect from a target and the echoes return to a sensing receiver with a uniform linear array modeled by Chu's CMS antennas. We assume small antenna spacing to have a TCA array. Wideband signals are used with the purpose of using the broad operational bandwidth of TCAs for sensing. We assume echoes returning to the sensing receiver are unknown and wideband. Thus, our approach is valid for any wideband waveform.

We denote the number of antennas at the sensing receiver by $N_R$. Since antennas at the sensing receiver are modeled as Chu's CMS antennas, the antenna array at the sensing receiver is modeled by RLC circuits. Thus, the signal received at the sensing receiver is represented by a voltage vector denoted by $\mathbf{v}_{k,t}$, with $k$ and $t$ being the $k$th frequency sample and $t$th time sample respectively. Define $K$ and $T$ as the total number of frequency samples and the total number of time domain samples respectively. The sensing receiver collects both time and frequency samples to detect the target. Thus, for $1 \!\leq \! k \! \leq \! K$ and $1 \!\leq \! t \! \leq \! T$, we present the detection problem with two hypotheses
\begin{equation}\label{Eq:originalDetProblem}
\begin{split}
    \mathcal{H}_0&: \mathbf v_{k,t} = \mathbf n_{k,t}\\    
    \mathcal{H}_1&: \mathbf v_{k,t} = \bm \alpha_{k,t}+\mathbf n_{k,t},
\end{split}
\end{equation}
where $\mathcal{H}_0$ and $\mathcal{H}_1$ represent the null and alternative hypotheses respectively, and vector $\bm \alpha_{k,t}$ indicates an unknown signal affected by the channel from the target to the sensing receiver. In Section \ref{Sec:NumericalResults}, a physically-consistent channel model is adopted. We assume that the signal $\bm\alpha_{k,t}$ is time varying and falls into three regimes in the frequency domain as constant, slowly or rapidly varying. We impose no prior assumption on the rate of time variation for $\bm\alpha_{k,t}$. In the absence of any knowledge about time variation, we cannot derive detectors built on temporal properties and treat each time sample separately. Also, $\mathbf n_{k,t}$ in \eqref{Eq:originalDetProblem} indicates the noise vector with a complex correlated Gaussian distribution with complex covariance matrix, $\mathbf{R}_k$, which is frequency dependent and derived in \cite[Eq. (16b)]{SWmassiveMIMO}. We assume uncorrelated noise in the time domain. Since both the channel and $\mathbf R_k$ model the effects of $a_R$, $\delta$ and MC, the problem in \eqref{Eq:originalDetProblem} is a physically-consistent detection problem.




We derive a general detection framework next. We will use it in later sections where we utilize estimates of the signal to derive detectors. We impose no assumption on the dynamics of the signal here. Let $\mathbf V$ be an array that has all collected samples, $\mathbf{v}_{k,t}$, in time and frequency. Using the PDF for a correlated Gaussian distribution and denoting the test threshold by $\gamma$, the likelihood ratio test for (\ref{Eq:originalDetProblem}) is written as 
\begin{equation}\label{Eq:LRT_movingAve1}
\Gamma_{\!\text{LRT}} (\mathbf{V}) \!\!=\!\!\frac{\exp \!\!\left(\!\sum_{k=1}^K \!\! \sum_{t=1}^T \!\! -(\!\mathbf{v}_{k,t} \!-\! \bm{\alpha}_{k,t}\!)^H \!\mathbf{R}^{-1}_k  (\!\mathbf{v}_{k,t} \!\! - \! \bm{\alpha}_{k,t}\!)\!\!\right)}{ \exp\!\! \left(\sum_{k=1}^K  \sum_{t=1}^T \! -\mathbf{v}_{k,t}^H \mathbf{R}^{-1}_k  \mathbf{v}_{k,t}\right)}\!\! \LRT{\mathcal{H}_1}{\mathcal{H}_0} \!\!\gamma\!,
\end{equation}
where $H$ denotes the conjugate transpose. Using (\ref{Eq:LRT_movingAve1}), the log-likelihood ratio (LLR) with a new threshold $\gamma_1$ is given as
\begin{equation}\label{Eq:logLikelihoodForProposedDet}
\begin{split}
&\Lambda_\text{LLR} (\mathbf{V}) \\&\!\!= \!\!\sum_{k=1}^K \!\! \sum_{t=1}^T \!\! \left(\!\mathbf{v}^H_{k,t}\mathbf{R}^{\!-1}_k\!\bm{\alpha}_{k,t}\!+\!\bm{\alpha}^H_{k,t}\mathbf{R}^{\!-1}_k\mathbf{v}_{k,t}
\!-\!\bm\alpha^H_{k,t}\mathbf{R}^{\!-1}_k\!\bm\alpha_{k,t} \! \right) \!\!\!\LRT{\mathcal{H}_1}{\mathcal{H}_0} \!\!\gamma_1\!.
\end{split}
\end{equation}
Next, we present our proposed approach and analysis.

\section{Moving Average Detector For a Slowly-Varying Signal in the Frequency Domain}\label{Sec:ProposedForAtleastOneSlow2}

We propose a novel detector for a time-varying signal that changes slowly over frequency, which we call the slowly-varying regime. The signal in \eqref{Eq:originalDetProblem} is unknown. Thus, a generalized likelihood ratio test (GLRT) can be utilized, where the signal in \eqref{Eq:logLikelihoodForProposedDet} is substituted with its MLE. In Section \ref{Sec:SignalWithoutSlowVar}, we will use the MLEs for the constant or rapidly-varying regimes. 

However, we take a different approach for the slowly-varying regime in this section, to exploit the constrained understanding about the slow variation of the signal. Our proposed approach is based on a moving average (MA). We use a window that slides in frequency to estimate the signal in \eqref{Eq:logLikelihoodForProposedDet} by uniformly averaging the received samples or voltage vectors that are within the window. Generally, the averaging can be performed non-uniformly but in this work we utilize uniform averaging. We show that the outputs of our MA detector under the null and alternative hypotheses are specific forms of a generalized chi-square variable. In Section \ref{Sec:NumericalResults}, we demonstrate the superior performance of our MA detector.

We now proceed to explain our MA approach in detail. We estimate signal $\bm\alpha_{k,t}$ in \eqref{Eq:logLikelihoodForProposedDet} with $\bm {\alpha}_{k,t}^\text{MA}$, where the superscript MA indicates our MA approach. We define $\bm {\alpha}_{k,t}^\text{MA}$ for frequency sample $k$ as a weighted average of voltage vectors inside a window, in the frequency domain, that covers $\mathbf{v}_{k,t}$ and its surrounding samples. Denote by $w_{k,\ell}$ the weights for the weighted average. Mathematically, we write $\bm {\alpha}_{k,t}^\text{MA}$ as
\begin{equation}\label{Eq:windowingFormulaMA1}
\bm {\alpha}_{k,t}^\text{MA}=\sum_{{\ell}=-{L}}^{L} \mathbf{v}_{{k-\ell},t} w_{k,\ell},   
\end{equation}
where $L$ determines the window length which is $2L+1$. We assume $L$ is even and the window length is smaller than total number of frequency samples $K$. The $k$-th frequency sample in the window determines the center of the window with $\ell=0$. 

We use uniform averaging to compute the weights, however the estimation in \eqref{Eq:windowingFormulaMA1} is general and non-uniform weights can be used. Weights $w_{k,\ell}$ are set to zero for $k-\ell<1$ or $k-\ell>K$ to mitigate edge effects. In these edge cases, the averaging applies to frequency samples in the range $1\leq k \leq K$. Also, when $\ell<-L$ or $\ell>L$ weights $w_{k,\ell}$ are set to zero. Due to the edge effects, there are three constraints for the weights, and we calculate them as follows:
\begin{equation}\label{Eq:Three_casesForWin}
w_{k,\ell}=
\begin{cases}
1/(2L+1) &k\geq L+1, k\leq K-L\\
1/(k+L) &k\leq L \\
1/(K-k+1+L) &k\geq K-L+1.
\end{cases}
\end{equation}
We replace $\bm\alpha_{k,t}$ in \eqref{Eq:logLikelihoodForProposedDet} with $\bm\alpha_{k,t}^\text{MA}$ given in (\ref{Eq:windowingFormulaMA1}) to write our MA detector for the slowly-varying regime as
\begin{equation}\label{Eq:DetectorForCase(2,3)FirstEq}
\begin{split}
&\Lambda_\text{MA} (\mathbf{V}) =\sum_{t=1}^{T} \Bigg[\sum_{k=1}^{K} \sum_{\ell=-{L}}^{L} \mathbf v_{k,t}^H \mathbf{R}^{-1}_k  \mathbf{v}_{k-\ell,t} w_{k,\ell} \\
&+ \sum_{k=1}^{K} \Bigg(\sum_{\ell=-L}^{L} w_{k, \ell} \mathbf{v}_{k-\ell,t}^H\Bigg) \mathbf{R}^{-1}_k \mathbf v_{k,t} 
\\ & -\!\!\! \sum_{k=1}^{K} \!\!\Bigg(\!\sum_{\ell=-L}^{L} \!\! w_{k,\ell} \mathbf{v}_{k-\ell,t}^H\Bigg) \! \mathbf{R}^{-1}_k \! \Bigg(\!\sum_{{\ell'}=-{L}}^{L} \!\!\! \mathbf{v}_{k-{\ell'},t} w_{k,{\ell'}} \Bigg) \Bigg],
\end{split}
\end{equation}
where the subscript MA indicates our MA detector.
After some algebra, we write $\Lambda_\text{MA} (\mathbf{V})$ in (\ref{Eq:DetectorForCase(2,3)FirstEq}) in a quadratic form as 
\begin{equation}\label{Eq:logLikelihoodMatrixFormMAD1}
\Lambda_\text{MA} (\mathbf{V}) \!= \!\!\sum_{t=1}^T \!\!
\begin{bmatrix}
\mathbf v_{1,t}^H \, \mathbf v_{2,t}^H \, ... \, \mathbf v_{K,t}^H 
\end{bmatrix}
\!\! \mathbf{A} \!\!
\begin{bmatrix}
\mathbf v^T_{1,t} \, \mathbf v^T_{2,t} ... \, \mathbf v^T_{K,t}\!
\end{bmatrix}^T\!\!,
\end{equation}
where $\mathbf{A}_{KN_R\times KN_R}$ is a matrix that has sub-matrices $\mathbf{A}_{rs}$ with dimensions $N_R$ by $N_R$, and indices $r$ and $s$ are in the ranges $1 \leq r \leq K$ and $1 \leq s \leq K$. Also, $(.)^T$ refers to the transpose operation. We calculate the sub-matrices as 
\begin{equation}\label{Eq:ExpressionOfArsForCase(2,3)}
\begin{split}
\mathbf{A}_{rs}\!= \!\!\mathbf{R}^{-1}_r w_{r,r-s}\!\! +\! w_{s,s-r} \mathbf{R}^{-1}_s
\!\!-\! \!\sum_{k=1}^K \!w_{k,k-r} \mathbf{R}^{-1}_k w_{k,k-s}.
\end{split}
\end{equation}

We now analyze our MA detector in the following theorem.

\begin{theorem}\label{Theorem:TheoremForMAD1}
The detector $\Lambda_\text{MA}(\mathbf{V})$ in (\ref{Eq:logLikelihoodMatrixFormMAD1}) under hypotheses $\mathcal{H}_1$ and $\mathcal{H}_0$ is a linear combination of independent non-central chi-squares and a linear combination of independent chi-squares respectively, as follows  
\begin{equation}\label{Eq:InTheoremMAD1ForPD}
\Lambda_\text{MA}|\mathcal{H}_1= \sum_{i=1}^{KN_R} \lambda_i \,\, \chi^2_{2T,i}\bigg(\sum_{t=1}^T \operatorname{Re}\{\mu_{i,t}\}^2+\operatorname{Im}\{\mu_{i,t}\}^2\bigg),
\end{equation}
\begin{equation}\label{Eq:InTheoremMAD1ForP_FA}
\Lambda_\text{MA}|\mathcal{H}_0= \sum_{i=1}^{KN_R} \lambda_i \,\, \chi^2_{2T,i},
\end{equation}
where $\chi^2_{A(i), i}(B(i))$ for $i=1,...,KN_R$ denote independent non-central chi-square variables with $A(i)$ degrees of freedom (DoFs) and non-centrality parameter $B(i)$. Also, $\lambda_i$ is defined above \eqref{Eq:muDeltaCforProposed1} and $\mu_{i,t}$ is the $i$-th element of vector $\bm \mu_t$ in \eqref{Eq:DefineBMMu_t1}.
\end{theorem}
\begin{proof}
Given $\mathcal{H}_1$, $\mathbf{v}_{k,t}$ in (\ref{Eq:originalDetProblem}) follows a Gaussian distribution with mean and covariance matrix $\bm \alpha_{k,t}$ and $\mathbf{R}_k$ respectively. Hence, we can express $\mathbf v_{k,t}$ as $\mathbf v_{k,t} = \bm{\alpha}_{k,t}+\mathbf{R}_k^\frac{1}{2}\tilde{\mathbf{v}}_{k,t}$ with $\tilde{\mathbf{v}}_{k,t}$ being a standard Gaussian vector. Using this expression in (\ref{Eq:logLikelihoodMatrixFormMAD1}), we can write $\Lambda_\text{MA} (\mathbf{V})$ under $\mathcal{H}_1$ as
\begin{equation}\label{Eq:DetectorH1forPD_FactroRs1}
\Lambda_\text{MA}|\mathcal{H}_1 \!\! = \!\! \sum_{t=1}^T \!\!
\begin{bmatrix}
 \tilde{\mathbf{b}}_{1,t}^H \, \tilde{\mathbf{b}}_{2,t}^H \, ... \, \tilde{\mathbf{b}}_{K,t}^H 
\end{bmatrix}
\!\! \widetilde{\mathbf{R}}
\mathbf{A}
\widetilde{\mathbf{R}}
\!\!\begin{bmatrix}
\tilde{\mathbf{b}}^T_{1,t} \, \tilde{\mathbf{b}}^T_{2,t} ... \, \tilde{\mathbf{b}}^T_{K,t}
\end{bmatrix}^T\!\!\!,
\end{equation}
where $\tilde{\mathbf{b}}_{k,t}=\mathbf{R}_k^{-\frac{1}{2}}\bm{\alpha}_{k,t}+\tilde{\mathbf{v}}_{k,t}$ and $\widetilde{\mathbf{R}}_{KN_R\times KN_R}$ is a block diagonal matrix with sub-matrices $\mathbf{R}_1^\frac{1}{2}$, $\mathbf{R}_2^\frac{1}{2}$, ..., $\mathbf{R}_K^\frac{1}{2}$ on the diagonal.
We use the eigenvalue decomposition to replace matrix $\widetilde{\mathbf{R}}\mathbf{A}\widetilde{\mathbf{R}}$ in (\ref{Eq:DetectorH1forPD_FactroRs1}) by $\bm{\Phi}\bm{\Delta}\bm{\Phi}^H$. Here, $\bm{\Phi}$ is unitary and $\bm{\Delta}$ is a diagonal matrix containing the eigenvalues of $\widetilde{\mathbf{R}}\mathbf{A}\widetilde{\mathbf{R}}$, $\lambda_i$, where $i=1,2,...,KN_R$. Therefore, (\ref{Eq:DetectorH1forPD_FactroRs1}) is written as
\begin{equation}\label{Eq:muDeltaCforProposed1}
\Lambda_\text{MA}|\mathcal{H}_1 \!\! =\! \sum_{t=1}^T  (\bm{\mu}_t^H \!+\! \mathbf c_t^H \!)\bm{\Delta}(\bm{\mu}_t \!+\! \mathbf c_t \!)\!=\!\!\sum_{i=1}^{KN_R}\!\!\lambda_i \!\! \sum_{t=1}^T |\mu_{i,t}+c_{i,t}|^2,
\end{equation}
where 
\begin{equation}\label{Eq:DefineBMMu_t1}
\bm{\mu}_t \! = \! \bm{\Phi}^H \!\! \begin{bmatrix} \!
(\mathbf{R}_1^{-\frac{1}{2}} \bm{\alpha}_{1,t})^T \,~ (\mathbf{R}_2^{-\frac{1}{2}} \bm{\alpha}_{2,t})^T ... \, (\mathbf{R}_K^{-\frac{1}{2}} \bm{\alpha}_{K,t})^{T} \!
\end{bmatrix}^T\!\!,
\end{equation}
and $\mathbf c_t=\bm{\Phi}^H \begin{bmatrix}
\tilde{\mathbf{v}}_{1,t}^T \, \tilde{\mathbf{v}}_{2,t}^T ... \, \tilde{\mathbf{v}}_{K,t}^T
\end{bmatrix}^T$ which is a standard Gaussian vector as matrix $\bm{\Phi}$ is unitary. Also, $\mu_{i,t}$ and $c_{i,t}$ in \eqref{Eq:muDeltaCforProposed1} are respectively the $i$-th element in $\bm \mu_t$ and the $i$-th element in $\mathbf{c}_t$. 

The summation over $t$ in (\ref{Eq:muDeltaCforProposed1}) is a non-central chi-square variable with $2T$ DoFs and non-centrality parameter $\sum_{t=1}^T (\operatorname{Re}\{\mu_{i,t}\}^2+\operatorname{Im}\{\mu_{i,t}\}^2)$. Therefore, (\ref{Eq:muDeltaCforProposed1}) is a linear combination of non-central chi-squares, which is given in (\ref{Eq:InTheoremMAD1ForPD}). 

For the null hypothesis, since there is no signal component, referring to \eqref{Eq:originalDetProblem}, we can calculate $\Lambda_\text{MA}|\mathcal{H}_0$ similarly by replacing $\mathbf v_{k,t}$ in (\ref{Eq:logLikelihoodMatrixFormMAD1}) with $\mathbf{R}_k^\frac{1}{2}\tilde{\mathbf{v}}_{k,t}$. Following the same approach presented for the alternative hypothesis above, we obtain (\ref{Eq:InTheoremMAD1ForP_FA}) and this concludes our proof.
\end{proof}

Referring to \eqref{Eq:InTheoremMAD1ForPD} and \eqref{Eq:InTheoremMAD1ForP_FA}, our MA detector given the alternative and null hypotheses, are specific forms of a generalized chi-square variable, since they are respectively linear combinations of non-central chi-square or chi-square variables. We use the toolbox developed in \cite{AbhranilComputeGeneralized,Abhranil2} to generate the cumulative distribution function (CDF) of a generalized chi-square variable for the numerical results.

\section{Analysis for Constant or Rapidly-Varying Signals in Frequency}\label{Sec:SignalWithoutSlowVar}

We derive MLE-based detectors and present our analysis for constant or rapidly-varying regimes, i.e., time-varying signals that are constant or rapidly varying in the frequency domain.

\subsection{Analysis for the Constant Regime}\label{Sec:ConstantDetector}

To have a signal model that is constant in the frequency domain but it is varying in the time domain, we drop index $k$ from $\bm{\alpha}_{k,t}$ and represent the signal by $\bm{\alpha}_{t}$. Define $\bm \beta = \begin{bmatrix}
\bm \alpha_1^T \, \bm \alpha_2^T \, ... \, \bm \alpha_T^T  
\end{bmatrix}^T$. Also, define $\mathbf{M}$ as a block diagonal matrix whose diagonal blocks are $\sum_k \mathbf R^{-1}_k$ or $\mathbf{M}=\text{blockdiag}\{\sum_k \mathbf R^{-1}_k,...,\sum_k \mathbf R^{-1}_k\}$, and let $\mathbf{u}=\begin{bmatrix}(\sum_k \mathbf R^{-1}_k\mathbf
v_{k,1})^T\, (\sum_k \mathbf R^{-1}_k\mathbf
v_{k,2})^T\,...\,(\sum_k \mathbf R^{-1}_k\mathbf
v_{k,T})^T\end{bmatrix}^T$. We rewrite (\ref{Eq:logLikelihoodForProposedDet}) as
\begin{equation}\label{Eq:ConstantRegimeFirstEqu1}
\begin{split}
\Lambda_{\text{C}} (\mathbf{V})\!\!=& 2 \!\operatorname{Re}\!\Big\{\!\sum_{t=1}^T \! \bm \alpha_t^H \! \sum_{k=1}^K \! \mathbf R^{-1}_k\mathbf v_{k,t} \!\Big\}\!-\!\!\!\sum_{t=1}^T \!\bm \alpha_t^H \! \Big(\!\sum_{k=1}^K \mathbf R^{-1}_k\!\Big) \bm \alpha_t \\= & 2\operatorname{Re}\{\bm \beta^H\mathbf{u}\}-\bm \beta^H \mathbf{M} \bm \beta,
\end{split}
\end{equation}
where the subscript $\text{C}$ indicates the constant regime. Since we assume the signal is unknown, we use the GLRT approach explained in Section \ref{Sec:ProposedForAtleastOneSlow2}. To write the GLRT, we compute the MLE of $\bm{\beta}$ as $\bm \beta_\text{MLE}=\mathbf M^{-1}\mathbf{u}$. This is equivalent to maximizing the numerator in \eqref{Eq:LRT_movingAve1} over all $\bm \alpha_{t}$ for $t=1,2,...,T$. Substituting the MLE into \eqref{Eq:ConstantRegimeFirstEqu1} gives the detector as 
\begin{equation}\label{Eq:LambdaStarConstantD1}
\Lambda^*_{\text{C}}(\mathbf V)\!\!=\!\!\mathbf u^H \mathbf{M}^{-1} \mathbf{u}\!=\!\!\!\sum_{t=1}^T \!\!\bigg(\sum_{k=1}^K \!\mathbf{v}_{k,t}^H \mathbf R^{-1}_k\!\!\bigg)\! \mathbf{S} \! \bigg(\sum_{k=1}^K \! \mathbf R^{-1}_k \mathbf{v}_{k,t}\!\!\bigg),
\end{equation}
where $\mathbf S=(\sum_k \mathbf R^{-1}_k)^{-1}$. We now present our result for the constant regime in the following theorem.

\begin{theorem}\label{Theorem:TheoremForConstantRegime}
The detector $\Lambda_{\text{C}}^*(\mathbf{V})$ in (\ref{Eq:LambdaStarConstantD1}) under hypotheses $\mathcal{H}_1$ and $\mathcal{H}_0$ is respectively a non-central chi-square variable and a chi-square variable as 
\begin{equation}\label{Eq:InTheoremConstantD1ForPD}
\Lambda_{\text{C}}^*|\mathcal{H}_1= \chi^2_{2 N_R T}\bigg(\sum_{t=1}^T \sum_{j=1}^{N_R} \operatorname{Re}\{q_{j,t}\}^2+\operatorname{Im}\{q_{j,t}\}^2\bigg),
\end{equation}
\begin{equation}\label{Eq:InTheoremConstantDForPFalse}
\Lambda_{\text{C}}^*|\mathcal{H}_0= \chi^2_{2 N_R T},
\end{equation}
where $\chi^2_{A}(B)$ denotes a non-central chi-square variable with $A$ DoFs and non-centrality parameter $B$. Also, $q_{j,t}$ is the $j$-th element of vector $\mathbf q_t$ in \eqref{Eq:InTheoremTWoforCnsDHArriveAtChi2}. 
\end{theorem}
\begin{proof}
We designed the detector for the constant regime assuming the signal is constant in frequency. In practice, detectors are utilized when the assumptions hold approximately but are not exact. Thus, to keep our analysis general, we consider a general signal model that changes in both frequency and time. Therefore, in our analysis we use the signal notation $\bm{\alpha}_{k,t}$ that includes the frequency index.

Referring to (\ref{Eq:originalDetProblem}) and given $\mathcal{H}_1$, we can replace $\mathbf v_{k,t}$ in (\ref{Eq:LambdaStarConstantD1}) with $\bm{\alpha}_{k,t}+\mathbf{R}_k^\frac{1}{2}\tilde{\mathbf{v}}_{k,t}$, where $\tilde{\mathbf{v}}_{k,t}$ is defined above \eqref{Eq:DetectorH1forPD_FactroRs1}. This gives (\ref{Eq:LambdaStarConstantD1}) as
\begin{equation}\label{Eq:InTheoremTWoforCnsDH1}
\Lambda^*_{\text{C}}|\mathcal{H}_1\!\!=\!\!\!\sum_{t=1}^T \!\!\bigg(\! \sum_{k=1}^K \bm{\alpha}_{k,t}^H \mathbf R^{-1}_k+\mathbf{f}_t^H\!\!\bigg)\! \mathbf{S} \! \bigg(\sum_{k=1}^K \! \mathbf R^{-1}_k \bm{\alpha}_{k,t}+\mathbf{f}_t\!\!\bigg),
\end{equation}
where $\mathbf{f}_t=\sum_k \mathbf R_k^{-\frac{1}{2}} \tilde{\mathbf{v}}_{k,t}$ and hence $\mathbf{f}_t$ is a Gaussian vector with zero mean and covariance matrix $\mathbf{S}^{-1}$. We replace $\mathbf{f}_t$ with $\mathbf{S}^{-\frac{1}{2}}\mathbf{w}_t$ where $\mathbf{w}_t$ denotes a standard Gaussian vector. Also, we replace $\mathbf{S}$ with $\mathbf{S}^{\frac{1}{2}}\mathbf{S}^{\frac{1}{2}}$. Thus, (\ref{Eq:InTheoremTWoforCnsDH1}) can be rewritten as
\begin{equation}\label{Eq:InTheoremTWoforCnsDHArriveAtChi2}
\Lambda_{\text{C}}^*|\mathcal{H}_1 \! =\! \sum_{t=1}^T  (\mathbf{q}_t \!+\! \mathbf w_t )^H(\mathbf{q}_t \!+\! \mathbf w_t ),
\end{equation}
where $\mathbf{q}_t=\mathbf{S}^{\frac{1}{2}}\sum_k \! \mathbf R_k^{-1} \bm{\alpha}_{k,t}$. We denote the $j$-th element of $\mathbf{q}_t$ by $q_{j,t}$. The $t$-th term in the summation in (\ref{Eq:InTheoremTWoforCnsDHArriveAtChi2}) is a non-central chi-square variable distributed as $\chi^2_{2 N_R}(\sum_{j=1}^{N_R} \operatorname{Re}\{q_{j,t}\}^2+\operatorname{Im}\{q_{j,t}\}^2)$. Consequently, (\ref{Eq:InTheoremTWoforCnsDHArriveAtChi2}) is the non-central chi-square variable given in \eqref{Eq:InTheoremConstantD1ForPD}.

To calculate $\Lambda^*_{\text{C}}(\mathbf V)$ given $\mathcal{H}_0$, we can replace $\mathbf v_{k,t}$ in (\ref{Eq:LambdaStarConstantD1}) with $\mathbf{R}_k^\frac{1}{2}\tilde{\mathbf{v}}_{k,t}$, as there is no signal component, and take a similar approach used above for $\mathcal{H}_1$. This gives \eqref{Eq:InTheoremConstantDForPFalse}.
\end{proof}

\subsection{Analysis for the Rapidly-Varying Regime}\label{Sec:BaselineWideBdetector}

For the regime with a time-varying signal that rapidly varies over frequency, we derive a detector which is an extended energy detector. Our extension takes into account correlated noise with a frequency-varying covariance matrix as well as joint frequency and time samples. Define vectors $\widehat{\bm \alpha}$ and $\widehat{\mathbf{v}}$ as
\begin{equation}\label{Eq:TwoVectorsAlphaAndVHATsforRapidR}
\begin{split}
&\widehat{\bm \alpha}\!\!=\!\!\!\begin{bmatrix}\!
\bm \alpha_{1,1}^T \, \bm \alpha_{2,1}^T \,\! ...\! \, \bm \alpha_{K,1}^T \, \bm \alpha_{1,2}^T \, \bm \alpha_{2,2}^T \,\! ... \!\, \bm \alpha_{K,2}^T\, ... \, \bm \alpha_{1,T}^T \, \bm \alpha_{2,T}^T \,\! ... \!\, \bm \alpha_{K,T}^T\!
\end{bmatrix}^{\!T},\\
&\widehat{\mathbf{v}}\!=\!\!\begin{bmatrix}\!
\mathbf{v}_{1,1}^T \, \mathbf{v}_{2,1}^T \,\! ...\! \, \mathbf{v}_{K,1}^T \, \mathbf{v}_{1,2}^T \, \mathbf{v}_{2,2}^T \,\! ... \!\, \mathbf{v}_{K,2}^T\, ... \, \mathbf{v}_{1,T}^T \, \mathbf{v}_{2,T}^T \,\! ... \!\, \mathbf{v}_{K,T}^T\!
\end{bmatrix}^{T}.
\end{split}
\end{equation}
In addition, define a block diagonal matrix $\widehat{\mathbf{R}}_{N_R K T \times N_R K T}$ with blocks $\mathbf{R}_1\!,\!\! \, \mathbf{R}_2,\!\! \, ...\! \, ,\!\mathbf{R}_K$ that repeat $T$ times on its diagonal. Therefore, $\mathbf {\widehat{R}}$ is given as
\begin{equation}\label{Eq:RhatForRapidRegime1}
\mathbf {\widehat{R}}=\text{blockdiag}\{\mathbf{R}_1,\mathbf{R}_2,...,\mathbf{R}_K,...,\mathbf{R}_1,\mathbf{R}_2,...,\mathbf{R}_K\}.
\end{equation}

To obtain the detector for the rapidly-varying regime, we rewrite  (\ref{Eq:logLikelihoodForProposedDet}), referring to \eqref{Eq:TwoVectorsAlphaAndVHATsforRapidR} and \eqref{Eq:RhatForRapidRegime1}, as
\begin{equation}
\Lambda_{\text{R}} \!(\mathbf{V}) \!\!=2\operatorname{Re}\{\widehat{\bm \alpha}^H\widehat{\mathbf{R}}^{-1}\widehat{\mathbf v}\}-\widehat{\bm \alpha}^H\widehat{\mathbf{R}}^{-1}\widehat{\bm \alpha}, 
\end{equation}
where subscript $\text{R}$ indicates that the detector is corresponding to the rapidly-varying regime. As the signal is unknown we use the GLRT approach and calculate the MLE of $\widehat{\bm \alpha}$ as $\widehat{\bm \alpha}_\text{MLE}=\widehat{\mathbf{R}}\widehat{\mathbf{R}}^{-1}\widehat{\mathbf{v}}=\widehat{\mathbf{v}}$, as done in Section \ref{Sec:ConstantDetector}. This gives
\begin{equation}\label{Eq:MaxLikelihoodBaselineDet1}
\Lambda_{\text{R}}^* \!(\mathbf{V}) \!\!=\widehat{\mathbf{v}}^H \widehat{\mathbf{R}}^{-1} \widehat{\mathbf{v}}=\sum_{k=1}^K \sum_{t=1}^T \mathbf{v}_{k,t}^H \mathbf{R}_k^{-1} \mathbf{v}_{k,t}.
\end{equation}

In the following theorem, we present the analysis for $\Lambda_{\text{R}}^* \!(\mathbf{V})$.

\begin{theorem}\label{Theorem:TheoremForRapidlyVaryingRegime}
The detector $\Lambda_{\text{R}}^*(\mathbf{V})$ in (\ref{Eq:MaxLikelihoodBaselineDet1}) under hypotheses $\mathcal{H}_1$ and $\mathcal{H}_0$ is respectively a non-central chi-square variable and a chi-square variable as follows 
\begin{equation}\label{Eq:BaselineDet1TheoremH_1}
\Lambda_{\text{R}}^*|\mathcal{H}_1\!\!=\!\! \chi^2_{2 N_R K T}\!\bigg(\!\sum_{t=1}^T\! \sum_{k=1}^K \!\sum_{j=1}^{N_R}\! \operatorname{Re}\{x_{j,k,t}\}^2\!+\!\operatorname{Im}\{x_{j,k,t}\}^2\!\!\bigg).
\end{equation}
\begin{equation}\label{Eq:BaselineDet1TheoremH0}
\Lambda_{\text{R}}^*|\mathcal{H}_0\!=\! \chi^2_{2 N_R K T},
\end{equation}
where $x_{j,k,t}$ is the $j$-th element of vector $\mathbf x_{k,t}$ in \eqref{Eq:BaselineDetPrrofFinalSummation}. 
\end{theorem}
\begin{proof}
Referring to (\ref{Eq:originalDetProblem}), under $\mathcal{H}_1$ we can substitute $\bm{\alpha}_{k,t}+\mathbf{R}_k^\frac{1}{2}\tilde{\mathbf{v}}_{k,t}$ for $\mathbf v_{k,t}$ in (\ref{Eq:MaxLikelihoodBaselineDet1}), where $\tilde{\mathbf{v}}_{k,t}$ is defined above \eqref{Eq:DetectorH1forPD_FactroRs1}. We also write $\mathbf{R}_k^{-1}$ in (\ref{Eq:MaxLikelihoodBaselineDet1}) as $\mathbf{R}_k^{-\frac{1}{2}}\mathbf{R}_k^{-\frac{1}{2}}$. Thus, (\ref{Eq:MaxLikelihoodBaselineDet1}) is written as
\begin{equation}\label{Eq:BaselineDetPrrofFinalSummation}
\Lambda^*_{\text{R}}|\mathcal{H}_1\!\!=\sum_{t=1}^T \sum_{k=1}^K (\mathbf{x}_{k,t} \!+\! \tilde{\mathbf{v}}_{k,t} )^H(\mathbf{x}_{k,t} \!+\! \tilde{\mathbf{v}}_{k,t}),
\end{equation}
where $\mathbf{x}_{k,t}=\mathbf{R}_k^{-\frac{1}{2}}\bm \alpha_{k,t}$. The $j$-th element of vector $\mathbf{x}_{k,t}$ is denoted by $x_{j,k,t}$. For fixed $k$ and $t$, a term in the summation in (\ref{Eq:BaselineDetPrrofFinalSummation}) is a non-central chi-square as $\chi^2_{2 N_R}(\sum_{j=1}^{N_R} \operatorname{Re}\{x_{j,k,t}\}^2+\operatorname{Im}\{x_{j,k,t}\}^2)$. This indicates that (\ref{Eq:BaselineDetPrrofFinalSummation}) is an addition of independent non-central chi-squares, which is the non-central chi-square variable written in \eqref{Eq:BaselineDet1TheoremH_1}.

Similarly, $\Lambda_{\text{R}}^*|\mathcal{H}_0$ is calculated by substituting $\mathbf{R}_k^\frac{1}{2}\tilde{\mathbf{v}}_{k,t}$ for $\mathbf v_{k,t}$ in (\ref{Eq:MaxLikelihoodBaselineDet1}), following a similar approach explained above for the alternative hypothesis. It follows that $\Lambda_{\text{R}}^* \!(\mathbf{V})$ given $\mathcal{H}_0$ is the chi-square variable we give in \eqref{Eq:BaselineDet1TheoremH0}.
\end{proof}

\section{Upper Bound for Detection Performance}

Unlike previous sections where we assumed an unknown signal, in this section we assume the signal is known, to derive an upper bound on the target detection performance. Here, (\ref{Eq:logLikelihoodForProposedDet}) is simplified and results in the upper bound detector $\Lambda_{\text{Up.B.}} \!(\mathbf{V})=\operatorname{Re}\{\sum_{k=1}^K \sum_{t=1}^T \bm{\alpha}^H_{k,t}\mathbf{R}^{\!-1}_k\mathbf{v}_{k,t}\}$.
To obtain $\Lambda_{\text{Up.B.}} \!(\mathbf{V})$ given $\mathcal{H}_1$, we replace $\mathbf v_{k,t}$ in $\Lambda_{\text{Up.B.}} \!(\mathbf{V})$ with $\bm{\alpha}_{k,t}+\mathbf{R}_k^\frac{1}{2}\tilde{\mathbf{v}}_{k,t}$, where $\tilde{\mathbf{v}}_{k,t}$ is defined above \eqref{Eq:DetectorH1forPD_FactroRs1}. This gives
\begin{equation}\label{Eq:upperboundDetectorUnderH1}
\Lambda_{\text{Up.B.}}| \mathcal{H}_1 \!=m_u+\!\operatorname{Re}\bigg\{\!\sum_{k=1}^K\! \sum_{t=1}^T \bm{\alpha}^H_{k,t}\mathbf{R}^{\!\!-\frac{1}{2}}_k\bm \tilde{\mathbf{v}}_{k,t}\!\bigg\}\!,
\end{equation}
where the known scalar $m_u=\!\! \operatorname{Re}\!\big\{\!\!\sum_{k=1}^K \!\sum_{t=1}^T\! \bm{\alpha}^H_{k,t}\mathbf{R}^{\!-1}_k\!\bm \alpha_{k,t}\!\big\}\!$. The second term in (\ref{Eq:upperboundDetectorUnderH1}) is a real Gaussian variable with zero mean and variance $v_u=\frac{1}{2}\sum_{k=1}^K \!\sum_{t=1}^T\! \bm{\alpha}^H_{k,t}\mathbf{R}^{\!-1}_k\!\bm \alpha_{k,t}$. Hence $\Lambda_{\text{Up.B}}| \mathcal{H}_1$ follows a real Gaussian distribution with mean $m_u$ and variance $v_u$. Similarly $\Lambda_{\text{Up.B}}| \mathcal{H}_0$ is calculated by replacing $\mathbf v_{k,t}$ in $\Lambda_{\text{Up.B}} \!(\mathbf{V})$ with $\mathbf{R}_k^\frac{1}{2}\tilde{\mathbf{v}}_{k,t}$, which gives a real Gaussian distribution with zero mean and variance $v_u$. Here, it is simple to find a closed form for $P_{D,u}$ in terms of $P_{FA,u}$, where $P_{D,u}$ and $P_{FA,u}$ indicate the probabilities of detection and false alarm for the upper bound detector respectively. We obtain $P_{D,u}$ as $P_{D,u}=Q(Q^{-1}(P_{FA,u}) - m_u/\sqrt{v_u} )$, where $Q(.)$ denotes the Q-function.

\section{Numerical Results}\label{Sec:NumericalResults}

In practice, signals in frequency cannot be exactly categorized as slow, constant or rapid. Hence, assumptions cannot hold exactly and detectors $\Lambda_\text{MA}(\mathbf{V})$, $\Lambda_{\text{C}}^*(\mathbf{V})$ and $\Lambda_{\text{R}}^*(\mathbf{V})$ can be used in regimes for which they are not originally designed. Thus, we examine the detectors in terms of robustness to frequency variation. We also show the accuracy of our analysis. 


For a single time sample, we denote the average received signal-to-noise ratio (SNR) across frequency for $K$ frequency samples by $\overline{\gamma}_K$. We consider tight coupling (TC) or weak coupling (WC), indicating large or small MC respectively, at the sensing receiver that has $N_R=32$ antennas. The ratio $\frac{\delta}{a_R}$ controls the MC hence for TC and WC ratio $\frac{\delta}{a_R}$ is set to the optimal value 1.932 and double the optimal value respectively \cite{SWmassiveMIMO}. Also, $\delta$ is set to 0.5 cm and we keep the array size constant thus we change $a_R$ to have TC or WC\cite{SWmassiveMIMO}. Results are shown for echo signals received from the end-fire direction that yields the largest bandwidth for colinear TCAs \cite{SWmassiveMIMO,Bandara_2025_multiuser}. Receiving echo signals from other directions reduces performance, with the lowest performance at broadside. All circuit theory parameters including the mutual and self impedance of the Chu's CMS antennas are set similarly to those in \cite{SWmassiveMIMO}. We adopt the following model for the unknown signal: 
\begin{equation}\label{Eq:unknownSignalForSimulation}
\bm \alpha_{k,t} = \mathbf{h}_{k,t} ge^{j(k-1)\theta_k}e^{j(t-1)\theta_t}  
\end{equation}
where $\mathbf{h}_{k,t}$ indicates the frequency-varying channel from the target to the sensing receiver, which is a physically-consistent line-of-sight channel\cite{SWmassiveMIMO}. Also $g$ denotes the echo signal which is modeled as a constant source voltage and calculated based on allocating the power equally to frequency samples for each time sample. Parameters $\theta_k$ and $\theta_t$ in (\ref{Eq:unknownSignalForSimulation}) model the rate of variation in the frequency and time domains respectively and increasing them causes larger frequency and time variations. In addition, we consider a small number of time samples taken in a short amount of time, which means that $\mathbf{h}_{k,t}$ is fixed in time.
\begin{figure}[t]
    \centering
    \includegraphics[width=0.8\linewidth]{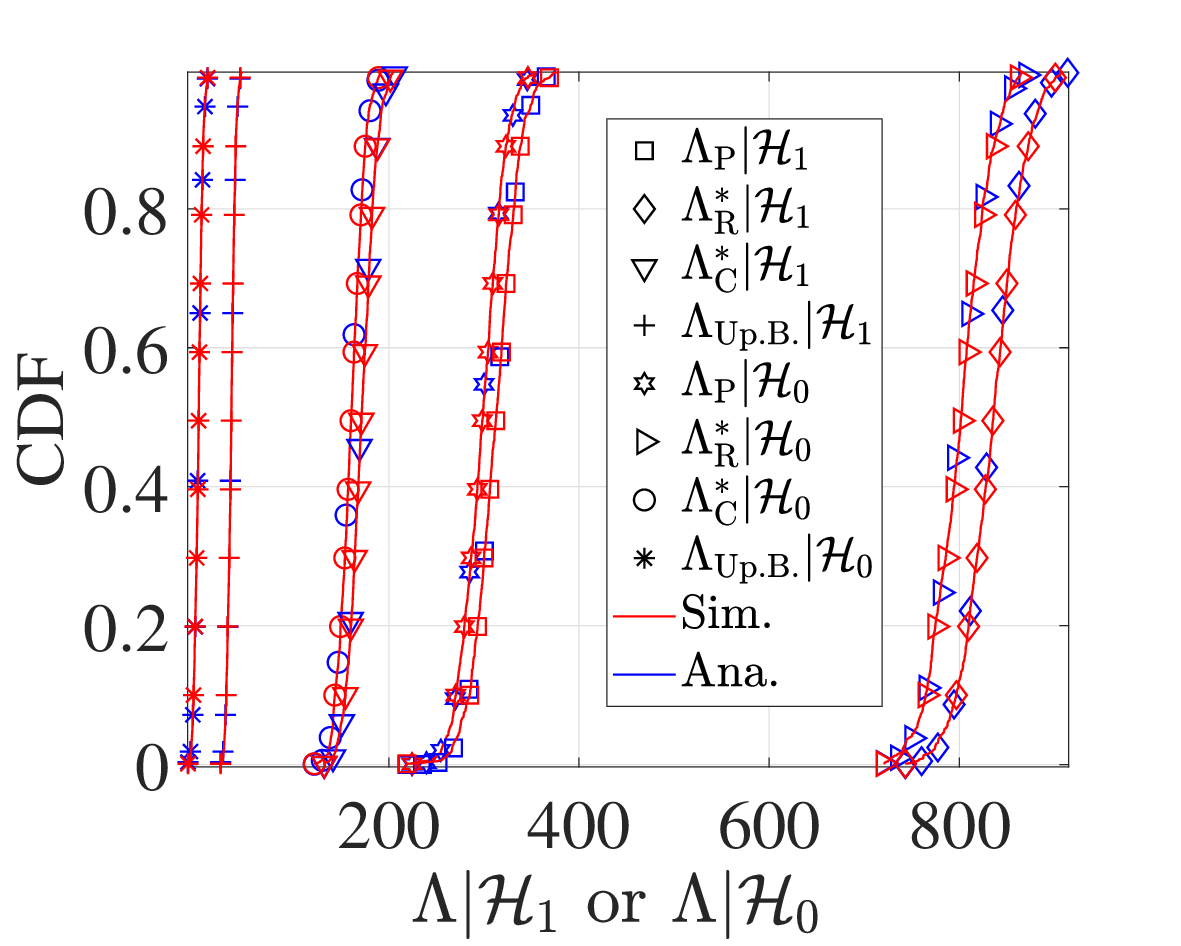}
    \caption{Distributions of the detectors for the alternative or null hypotheses for TC, bandwidth [20-30] GHz, $K=T=5$, $\theta_k$ = 0.2, $\theta_t$ = 0.2 and $\overline{\gamma}_K$ = 1.4 dB. The window length for our MA detector is 3.}\label{Fig:Distrib_CompareAnaSim}
\end{figure}
\begin{figure}[t]
    \centering
    \includegraphics[width=0.8\linewidth]{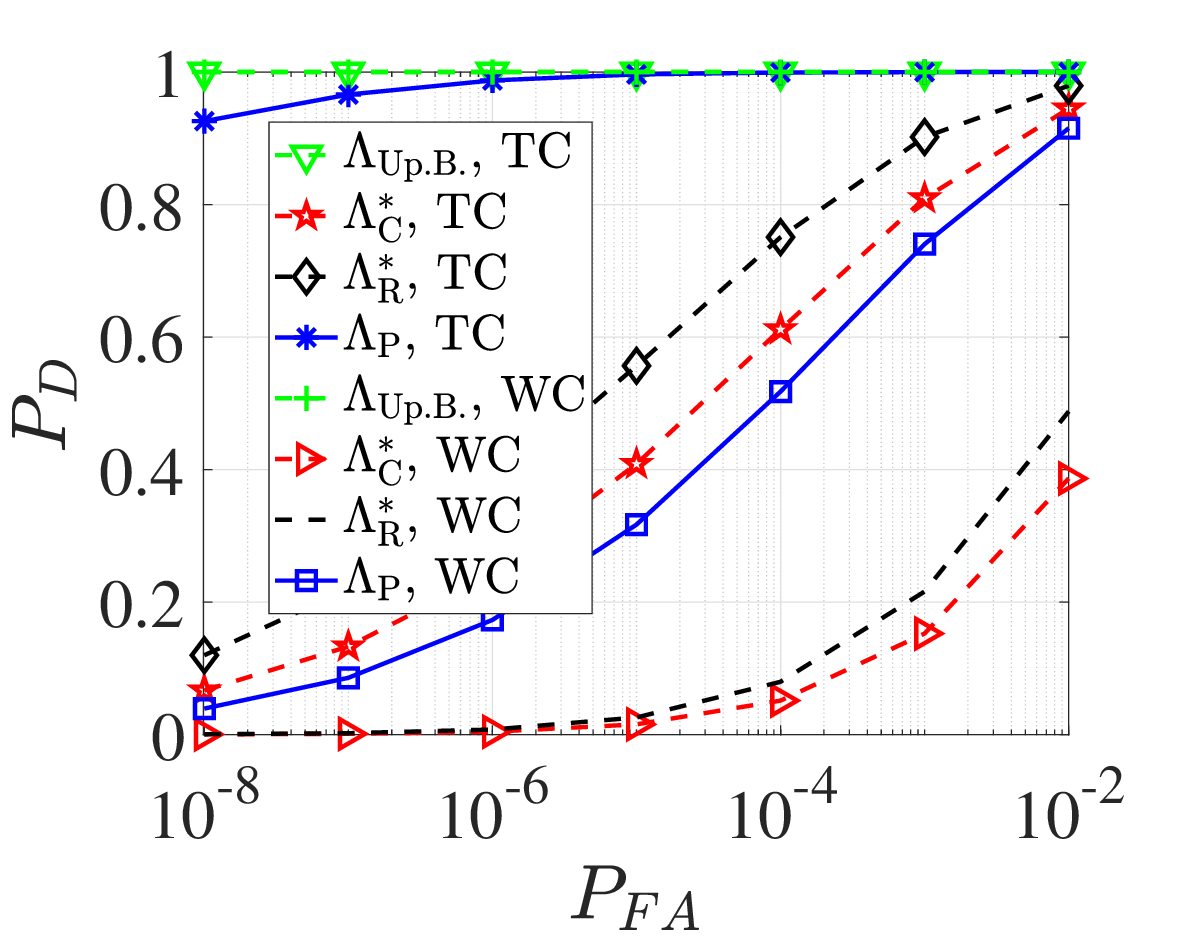}
    \caption{Probability of detection VS. probability of false alarm for arrays with TC and WC with $\theta_t=0.5$ and $\theta_k=0.1$. The bandwidth is [28-30] GHz, $K=15$ and $T=8$. $\overline{\gamma}_K$ for arrays with TC and WC is 3.7 dB and 0.8 dB respectively, for the same total power. The window length for our MA detector is 5.}\label{Fig:PD_versus_PAF1}
\end{figure}
\begin{figure}[t]
    \centering
    \includegraphics[width=0.8\linewidth]{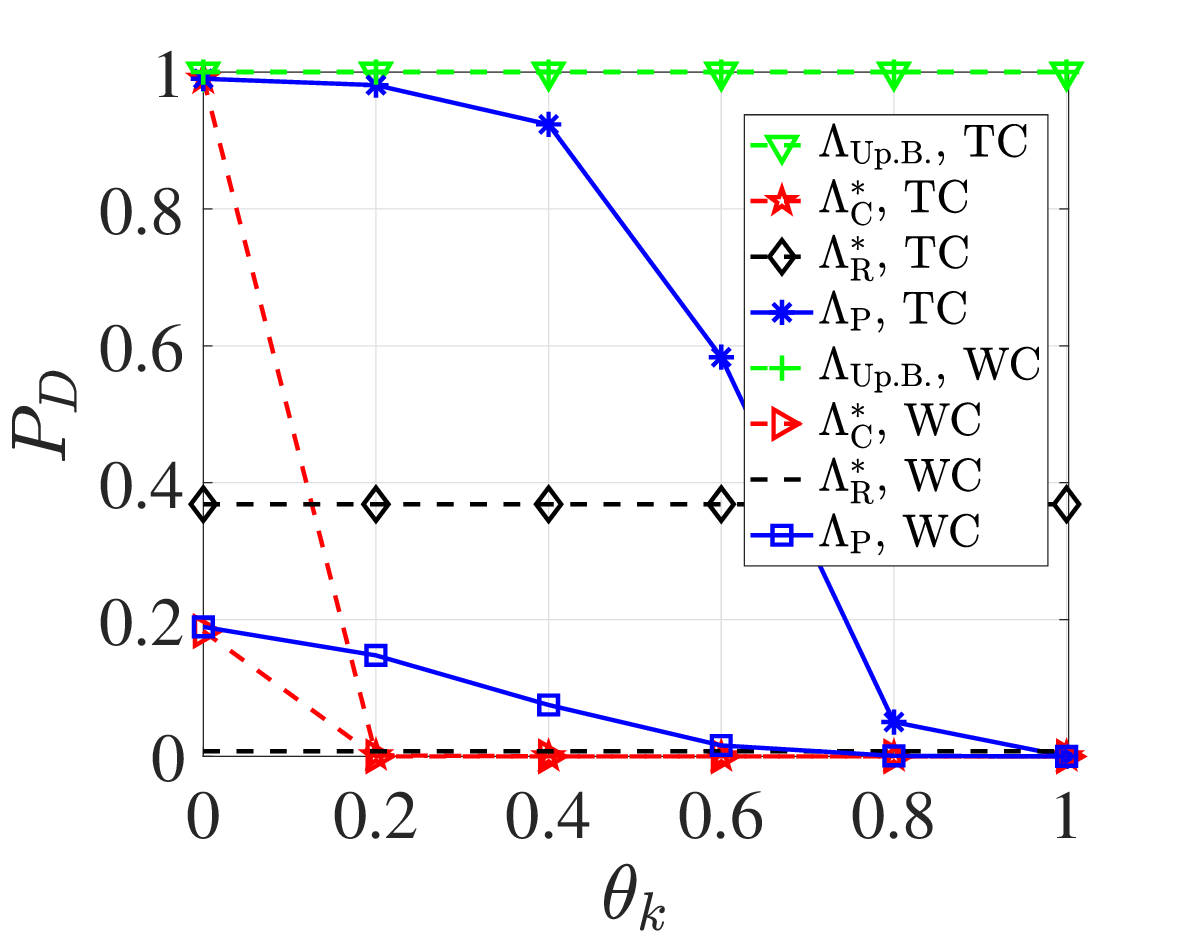}
    \caption{Probability of detection VS. rate of variation in frequency for arrays with TC and WC with $\theta_t=0.5$. The bandwidth is [28-30] GHz, $P_{FA}=10^{-6}$, $K=15$ and $T=8$. The average of $\overline{\gamma}_K$ for arrays with TC and WC is 3.7 dB and 0.8 dB respectively, for the same total power. The window length for our MA detector is 5.}\label{Fig:PD_versus_ThetaK}
\end{figure}
Fig. \ref{Fig:Distrib_CompareAnaSim} compares our analysis with simulations of the CDF of the detectors, under the alternative or null hypotheses. Blue markers in this figure represent our analysis and red solid lines with the markers represent simulations. The figure confirms the complete accuracy of our derived distributions that perfectly agree with the simulated distributions.

Fig. \ref{Fig:PD_versus_PAF1} compares detection performance of the three detectors over a wide bandwidth of 2 GHz for antenna arrays with TC and WC. The figure shows the probability of detection, $P_D$, against the probability of false alarm, $P_{FA}$, when the signal is time varying with $\theta_t=0.5$ and slowly varying in frequency with $\theta_k=0.1$. For the calculations of $P_D$ and $P_{FA}$, we refer the reader to \cite{kay1998_2ndVolume}. The total power is set so that $\Lambda_{\text{Up.B.}}$ gives $P_D=1$ indicating that the other detectors are not limited by power. The figure confirms the superior performance of our MA detector $\Lambda_\text{MA}$ compared to other detectors including the extended energy detector, $\Lambda^*_\text{R}$. The figure also shows that TC results in larger $P_D$ compared to WC. Although the total power of the echoes is the same for both the arrays with TC and WC, TC leads to a larger average SNR over the wide bandwidth, which yields better performance.


Fig. \ref{Fig:PD_versus_ThetaK} shows detection performance versus the rate of variation of the signal over frequency, for $\theta_t=0.5$. It can be seen that our MA detector $\Lambda_\text{MA}$ outperforms the other detectors, particularly when the signal is slowly varying in frequency. In the figure, the extended energy detector $\Lambda^*_\text{R}$ shows fixed $P_D$ since in energy detection the frequency variation of the signal is not taken into account. The channel $\mathbf{h}_{k,t}$ is almost constant in frequency over the considered 2 GHz bandwidth. As a result, for very small values of $\theta_k$ the signal is almost fixed in frequency and the detector $\Lambda^*_\text{C}$ performs well, as it is designed for a constant signal over frequency. Notably, Fig. \ref{Fig:PD_versus_ThetaK}, similar to Fig. \ref{Fig:PD_versus_PAF1}, illustrates the advantage of TC
over WC.

\section{Conclusion}
We presented analysis for target detection using TCAs, and proposed a novel detector for time-varying signals with slow variation over frequency. We conducted performance analysis for detectors designed for other regimes with constant or rapidly-varying signals. We also derived and analyzed a detector which gives an upper bound on performance. We showed that our proposed detector exhibits greater robustness to frequency variation compared to other detectors. Additionally, we highlighted that TCAs show clear superiority over weakly-coupled antennas in target detection. This is due to the larger average SNR that TCAs provide over a wide bandwidth, at the cost of greater antenna design complexity.



\bibliographystyle{IEEEtran}
\bibliography{References}
\end{document}